\documentclass[11pt]{article}
\usepackage{fullpage}

\usepackage{hyperref}
\usepackage[utf8]{inputenc} %
\usepackage[T1]{fontenc}    %
\usepackage{url}            %
\usepackage{mathtools}
\usepackage{booktabs}       %
\usepackage{amsfonts}       %
\usepackage{nicefrac}       %
\usepackage{microtype}      %
\usepackage{xcolor}         %
\usepackage{MnSymbol}
\usepackage{xspace}

\usepackage{tikz}
\usepackage{pgfplots}
\usetikzlibrary{calc, arrows.meta}
\pgfplotsset{compat=1.18}

\usepackage{pict2e,picture,graphicx}
\usepackage[overload]{empheq}
\usepackage{microtype}
\usepackage{graphicx}
\usepackage{subcaption}
\usepackage{booktabs}
\usepackage{amsmath}
\usepackage{cases}
\usepackage{mathtools}
\usepackage{amsthm}
\usepackage{thm-restate}
\usepackage{enumitem}

\allowdisplaybreaks
\usepackage{xparse}
\usepackage[capitalize, noabbrev]{cleveref}
\usepackage{wrapfig}
\usepackage{bbm}
\usepackage{yfonts}
\usepackage{nicefrac} 
\usepackage{multirow}
\usepackage{bm}
\usepackage{bbm}
\usepackage{subcaption}

\hypersetup{
	colorlinks = true,
	linkcolor = black,
	citecolor = black,
	linktocpage = true,
	urlcolor = black
}
\usepackage{comment}

\newcommand{\CFont}[1]{{\textup{\textsf{#1}}}\xspace}
\newcommand{\PPAD}{\CFont{PPAD}}
\newcommand{\Pclass}{\CFont{P}}
\newcommand{\PLS}{\CFont{PLS}}
\newcommand{\CLS}{\CFont{CLS}}
\newcommand{\NP}{\CFont{NP}}
\newcommand{\NOR}{\textit{NOR}}
\newcommand{\PURIFY}{\textit{PURIFY}}

\newcommand{\TFNP}{\CFont{TFNP}}
\newcommand{\EOTL}{\CFont{End-of-the-Line}}

\newcommand{\GDAFP}{\CFont{GDA}}

\newcommand{\LINVI}{\CFont{LinVI}}
\newcommand{\PURECIRC}{\CFont{Pure-Circuit}}

\newcommand{\GDA}{\CFont{GDA}}

\newcommand{\Gcal}{\mathcal{G}}

\newcommand{\poly}{\textsf{poly}}

\newcommand{\Naturals}{\mathbb{N}}

\theoremstyle{plain}

\makeatletter
\let\cref@old@stepcounter\stepcounter
\def\stepcounter#1{%
  \cref@old@stepcounter{#1}%
  \cref@constructprefix{#1}{\cref@result}%
  \@ifundefined{cref@#1@alias}%
    {\def\@tempa{#1}}%
    {\def\@tempa{\csname cref@#1@alias\endcsname}}%
  \protected@edef\cref@currentlabel{%
    [\@tempa][\arabic{#1}][\cref@result]%
    \csname p@#1\endcsname\csname the#1\endcsname}}
\makeatother

\theoremstyle{plain}
\newtheorem{theorem}{Theorem}[section]

\newtheorem{lemma}[theorem]{Lemma}

\newtheorem*{theorem-non}{Theorem}

\theoremstyle{definition}

\newtheorem{problem}{Problem}

\theoremstyle{remark}
\newtheorem{remark}[theorem]{Remark}

\usepackage[style=alphabetic,natbib=true,maxcitenames=2, maxbibnames=10,backend=bibtex]{biblatex}
\allowdisplaybreaks

\title{The Complexity of Min-Max Optimization with Product Constraints}

\author{
 Martino Bernasconi \\
 Bocconi University\\
 {\textcolor{black}{\small\texttt{martino.bernasconi@unibocconi.it}}}
 \and 
 Matteo Castiglioni  \\
 Politecnico di Milano\\
 {\textcolor{black}
 {\small\texttt{matteo.castiglioni@polimi.it}}}
}

\date{}

\begin{document}

\maketitle

\begin{abstract}
    We study the computational complexity of the problem of computing local min-max equilibria of games with a nonconvex-nonconcave utility function $f$. From the work of \citet*{daskalakis2021complexity}, this problem was known to be hard in the restrictive case in which players are required to play strategies that are jointly constrained, leaving open the question of its complexity under more natural constraints. In this paper, we settle the question and show that the problem is \PPAD-hard even under product constraints and, in particular, over the hypercube. 
\end{abstract}

\pagenumbering{gobble} 

\clearpage
\newpage
\pagenumbering{arabic}

\section{Introduction}

This work studies the fundamental problem of finding local min-max points of a continuous and smooth function
\[
f:[0,1]^d\times[0,1]^d\to[0,1].
\]

The computation of saddle points is a fundamental primitive in mathematical optimization, with important practical applications to many problems,
ranging from training generative adversarial networks (GANs) \citep{goodfellow2014generative, arjovsky2017wasserstein} and adversarial training \citep{madry2017towards, sinha2018certifying} to alignment \citep{munos2024nash}.

In the last decade, the importance of such problems has motivated the theoretical study of algorithms that converge to such solutions. For instance, in the nonconvex-concave setting, many algorithms converge at polynomial rates \citep{ostrovskii2021efficient, jin2021nonconvex, lin2020gradient, thekumparampil2019efficient, lin2025two}. For the more general nonconvex-nonconcave setting, results are limited to local convergence \citep{mertikopoulos2018optimistic, wang2019solving, balduzzi2018mechanics, mazumdar2025finding} or require additional structure \citep{diakonikolas2021efficient, abernethy2021last, kalogiannis2024learning}.

Despite its importance, the exact algorithmic characterization of this problem remains elusive. In particular, no known algorithms provably find $\epsilon$-approximate solutions in time $\poly(1/\epsilon)$.

This theoretical gap mirrors a well-documented frustration in practice, particularly when training GANs and adversarial systems \citep{xing2021algorithmic, kodali2017convergence, hsieh2021limits, vlatakis2019poincare}. The difficulty is often attributed to the non-convexity of local solutions to min-max problems. However, this explanation feels unsatisfactory since finding stationary points in non-convex minimization is quite tractable (see also the discussion in \cite{daskalakis2021complexity}). In this paper, we prove a more fundamental reason for this failure: finding local min-max points is computationally hard. In particular, we show that $\poly(1/\epsilon)$ algorithms cannot exist unless $\PPAD=\Pclass$.

More formally, we show that the following problem called \GDA (Gradient Descent-Ascent, see \Cref{def:gdafp}), is \PPAD-hard: given a differentiable function $f:[0,1]^{d}\times[0,1]^d\to[0,1]$, find $(x,y)$ such that
\[
\max_{x'\in[0,1]^d}\langle\nabla_xf(x,y),x'-x\rangle\le\epsilon\quad\text{and}\max_{y'\in[0,1]^d}-\langle\nabla_yf(x,y),y'-y\rangle\le\epsilon.
\]
This problem requires finding a point $(x,y)$ where both the max player $x$ and the min player $y$ are at a local optimum, or equivalently, a fixed point of the gradient-descent-ascent dynamic.

This problem was first analyzed through the lens of computational complexity by \citet{daskalakis2021complexity}, who proved that it lies in \PPAD. They also proved \PPAD-hardness, but only for the case where the $x$ and $y$ players are restricted to play in a coupled (i.e., non-product) subset of the hypercube $[0,1]^{2d}$.

As observed in many subsequent works \citep{babichenko2021settling, fearnley2022complexity, bernasconi2024roleconstraintscomplexityminmax, hollender2024complexity, anagnostides2025complexity}, characterizing the complexity of min-max optimization in the more natural setting of product constraints has remained a fundamental open problem.
Recently, \citet{bernasconi2024roleconstraintscomplexityminmax} simplified the reduction by \citet{daskalakis2021complexity}, improving and generalizing the result while also emphasizing the central role of coupled constraints in the reduction.

In this paper, we finally resolve the open problem and prove that finding local min-max points is \PPAD-hard, even under product constraints and in particular over the hypercube.

\paragraph*{Computational Complexity}

From a formal standpoint, \GDAFP is a total search problem in \NP; that is, a problem in \NP that always has a solution. \citet{megiddo1991total} shows that this class of problems, \TFNP, is unlikely to have any complete problems. \citet{papadimitriou1994complexity} subsequently defined subclasses of \TFNP, each corresponding to a nonconstructive (and inefficient) proof principle guaranteeing totality. The most relevant to our work is \PPAD, which relies on a parity argument on a directed graph. This class has been found to be the right one for various fixed-point problems, and most famously, Nash equilibria \citep{daskalakis2009complexity, chen2009settling}.

Recently, \citet{fearnley2021complexity, fearnley2022complexity} proved that finding a fixed point of gradient descent of a smooth and continuous function is complete for the class $\CLS=\PPAD\cap\PLS$. This shows that GD is simultaneously a fixed-point problem and a local-search problem. In particular, it is easy to see that the problem of finding a fixed point of gradient descent can be reduced to the following (informal) problem (which was used to define \CLS in \citep{daskalakis2011continuous}): given a continuous ``potential function'' $p:[0,1]^d\to[0,1]$ and a continuous ``improvement function'' $g:[0,1]^d\to[0,1]^d$, find a $x$ such that $p(g(x))\ge p(x)-\epsilon$. Clearly, fixed points of $g$ are solutions to this problem. 
In this paper, we prove that, unlike gradient descent, gradient descent-ascent (GDA) is a pure fixed-point problem without a potential function. This aligns with the well-documented phenomenon observed in min-max problems, where algorithms naturally tend to cycle.

\subsection*{Technical Challenges}

The recent hardness results for the computation of Nash equilibria are driven by the idea of \emph{imitation} \citep{babichenko2021settling, babichenko2016query, rubinstein2017settling, rubinstein2015inapproximability}. At a high level, given a vector-valued function $F$ of which fixed points are hard to find, the $x$ player (fixed-point player) minimizes the distance between $x$ and $F(y)$, and the $y$ player (imitation player) minimizes the distance between $x$ and $y$.  However, this approach requires different objectives for the $x$ and $y$ players. The fundamental challenge in min-max problems is that we are forced to use opposite objectives between the two players. Thus, forcing one player to imitate the other would compel the other to escape. This is one of the main conceptual obstacles to overcome.

Intuitively, we give up trying to enforce $x\simeq y$. Instead, our construction combines two \PPAD-hard problems: \LINVI \citep{bernasconi2024roleconstraintscomplexityminmax}, which is hard when $x\simeq y$, and \PURECIRC \citep{deligkas2022pure}.
In particular, we allow \PURECIRC gadgets to fail, but, in that case, the construction forces $x\simeq y$ and thus a solution to the \LINVI problem, which is also \PPAD-hard.
In the next section, we provide a high-level technical overview of our proof.

\section{Proof Overview}\label{sec:overview}

All known reductions for the coupled constrained setting \citep{daskalakis2021complexity,bernasconi2024roleconstraintscomplexityminmax,anagnostides2025complexity} use a utility function of the form $f(x,y)=\langle F(x),x-y\rangle$.
In particular, \citet{bernasconi2024roleconstraintscomplexityminmax} made explicit the simple connection between min-max optimization with coupled constraints and Variational Inequalities (VIs). A VI is the problem of finding a point $z$ such that $\langle F(z),z'-z\rangle\ge 0$ for all $z'\in[0,1]^d$. It is well-known that many choices of the function $F$, e.g., linear, make the associated variational inequality \PPAD-hard.
Taking $F$ as a VI operator, combined with the coupled constraint $x\simeq y$, is sufficient to prove \PPAD-hardness under coupled constraints. Indeed, we can observe that $\nabla_x f(x,y)=F(x)+J_F(x)(x-y)$ and $\nabla_y f(x,y)=-F(x)$. The ``noise term'' $\Delta(x,y)=J_F(x)(x-y)$ is negligible under the coupled constraints, rendering the problem essentially equivalent to the original VI instance.\footnote{In this section, we use the symbol $\Delta$ quite informally to denote quantities that we would like to keep small.} Keeping the term $\Delta(x,y)$ in $\nabla_x f(x,y)$ under control is precisely the purpose of the $x\simeq y$ constraint.

The main challenge in getting rid of the coupled constraint is finding an alternative means of dealing with the noise component $\Delta(x,y)$. The following discussion aims to give an intuitive overview of our techniques.

\subsection*{\PURECIRC}

We start by considering a different operator $F$. Specifically, we reduce from \PURECIRC \citep{deligkas2022pure}, a problem that involves finding an assignment to nodes in a graph that satisfies a set of gates (constraints), which impose restrictions on the values of the vertices. The \PURECIRC problem is particularly nice as: (1) it works on continuous values $[0,1]$ but only cares about pure values $\{0,1\}$, treating intermediate values $(0,1)$ as ``garbage'' (denoted by $\bot$); and (2) its dependency graph is very structured. This allows us to reason concretely about the influence between variables and deal with the $\Delta(x,y)$ component.

For each node $v$ in the \PURECIRC instance, we introduce a local game played by two variables $x^v,y^v$ and interpret the assignment to $v$ as $1$ if $(x^v-y^v)^2$ is large, and $0$ if $(x^v-y^v)^2$ is small. For each output node $q$ of a gate, it is not difficult to implement a smooth function $s_q(x,y)$ that computes the numerical values that should be assigned to the output of that gate. In particular, $s_q(x,y)$ depends only on the values $(x^v-y^v)^2$ for nodes $v$ that are inputs to $q$'s gate.

The challenge is to ensure consistency: we must enforce that, at solutions of our instance, the values $(x^q-y^q)^2$ on output nodes $q$ are consistent with the values computed by $s_q(x,y)$. In this language, consistency means that if $s_q(x,y)= 1$, then $(x^q-y^q)^2$ must be large, and if $s_q(x,y)=0$, then $(x^q-y^q)^2$ must be small.

\subsection*{The Blueprint}
Here, we outline an initial attempt at a reduction. We define the utility function as 
\begin{align}\label{eq:firstutil}
f(x,y)=\sum_{q\in V}s_q(x,y)H_q(x,y),
\end{align}
where, crucially, $s_q(x,y)$ is the function computing the value of the gate of which $q$ is the output (and thus depends only on the variables at the inputs of node $q$), and $H_q(x,y)$ is a ``linking function'' that depends only on the variables related to node $q$ itself (the terminology for $H_q$ will be explained shortly).

We can write $f$ to better highlight the dependence on the variables $(x^q,y^q)$ for a fixed node $q$. Indeed:
\begin{align}\label{eq:decompos_intro}
f(x,y) = \underbrace{\vphantom{\sum_{\tilde u:\text{ $q$ is input to $\tilde u$}}}s_q(x,y)H_q(x,y)}_{(A)}+\underbrace{\sum_{\tilde u:\text{ $q$ is input to $\tilde u$}} s_{\tilde u}(x,y)H_{\tilde u}(x,y)}_{(B)}+\underbrace{\sum_{\text{all other nodes $u$}} s_u(x,y)H_u(x,y)}_{(C)}.
\end{align}

In $(C)$, the variables $(x^q,y^q)$ are not present, while in $(B)$ they are present only in $s_{\tilde u}$ and in $(A)$ only in $H_q(x,y)$. That is why we named $H_q$ the ``linking function'', since without it, the partial derivatives of $f$ with respect to $(x^q,y^q)$ would not depend on $s_q$. In particular, to maintain consistency, we want the players' incentives at node $q$ to depend only on the variables from the nodes to which $q$ is an output.

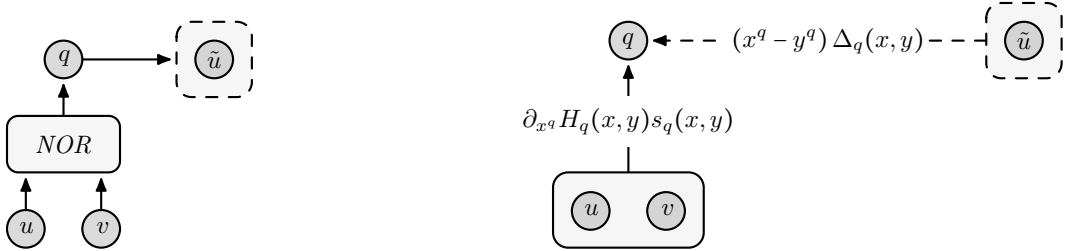
\begin{figure}
\centering
\begin{subfigure}{.4\textwidth}
  \centering
  \scalebox{0.95}{       
\tikzset{
  >=Latex[round],
  every picture/.style={line width=0.85pt,line cap=round,line join=round},
  every node/.style={font=\small},
}

\begin{tikzpicture}[x=0.75pt,y=0.75pt,yscale=-1,xscale=1]

\draw [fill=gray!7,fill opacity=1] (230,166) .. controls (230,162.69) and (232.69,160) .. (236,160) -- (284,160) .. controls (287.31,160) and (290,162.69) .. (290,166) -- (290,184) .. controls (290,187.31) and (287.31,190) .. (284,190) -- (236,190) .. controls (232.69,190) and (230,187.31) .. (230,184) -- cycle ;
\draw [->]   (240,210) -- (240,193) ;
\draw  [->]  (280,210) -- (280,193) ;
\draw [->]   (260,160) -- (260,143) ;
\draw  [fill={rgb, 255:red, 155; green, 155; blue, 155 }  ,fill opacity=0.37 ] (270,220) .. controls (270,214.48) and (274.48,210) .. (280,210) .. controls (285.52,210) and (290,214.48) .. (290,220) .. controls (290,225.52) and (285.52,230) .. (280,230) .. controls (274.48,230) and (270,225.52) .. (270,220) -- cycle ;
\draw  [fill={rgb, 255:red, 155; green, 155; blue, 155 }  ,fill opacity=0.37 ] (230,220) .. controls (230,214.48) and (234.48,210) .. (240,210) .. controls (245.52,210) and (250,214.48) .. (250,220) .. controls (250,225.52) and (245.52,230) .. (240,230) .. controls (234.48,230) and (230,225.52) .. (230,220) -- cycle ;
\draw [fill=gray!7,fill opacity=1, dash pattern={on 4.5pt off 4.5pt}] (320,118) .. controls (320,113.58) and (323.58,110) .. (328,110) -- (352,110) .. controls (356.42,110) and (360,113.58) .. (360,118) -- (360,142) .. controls (360,146.42) and (356.42,150) .. (352,150) -- (328,150) .. controls (323.58,150) and (320,146.42) .. (320,142) -- cycle ;
\draw  [->]  (270,130) -- (317,130) ;
\draw  [fill={rgb, 255:red, 155; green, 155; blue, 155 }  ,fill opacity=0.37 ] (330,130) .. controls (330,124.48) and (334.48,120) .. (340,120) .. controls (345.52,120) and (350,124.48) .. (350,130) .. controls (350,135.52) and (345.52,140) .. (340,140) .. controls (334.48,140) and (330,135.52) .. (330,130) -- cycle ;
\draw  [fill={rgb, 255:red, 155; green, 155; blue, 155 }  ,fill opacity=0.37 ] (250,130) .. controls (250,124.48) and (254.48,120) .. (260,120) .. controls (265.52,120) and (270,124.48) .. (270,130) .. controls (270,135.52) and (265.52,140) .. (260,140) .. controls (254.48,140) and (250,135.52) .. (250,130) -- cycle ;

\draw (260,175) node   [align=left] {\NOR};
\draw (240.5,219.5) node    {$u$};
\draw (281,219.5) node    {$v$};
\draw (340,130) node    {$\tilde{u}$};
\draw (260.5,129.5) node    {$q$};

\end{tikzpicture}}
  \caption{An example of a simple \PURECIRC instance. In the dashed box we have all other nodes $\tilde u$ and gates of the circuit.}
  \label{fig:sub1}
\end{subfigure}%
\hfill
\begin{subfigure}{.55\textwidth}
  \centering
  \scalebox{0.95}{       
\tikzset{
  >=Latex[round],
  every picture/.style={line width=0.85pt,line cap=round,line join=round},
  every node/.style={font=\small},
}

\begin{tikzpicture}[x=0.75pt,y=0.75pt,yscale=-1,xscale=1]

\draw [->] (290,190) -- (290,133) ;
\draw  [->, dash pattern={on 4.5pt off 4.5pt}]  (480,120) -- (390,120) -- (303,120) ;
\draw  [fill=gray!7,fill opacity=1, dash pattern={on 4.5pt off 4.5pt}] (480,108) .. controls (480,103.58) and (483.58,100) .. (488,100) -- (512,100) .. controls (516.42,100) and (520,103.58) .. (520,108) -- (520,132) .. controls (520,136.42) and (516.42,140) .. (512,140) -- (488,140) .. controls (483.58,140) and (480,136.42) .. (480,132) -- cycle ;
\draw  [fill={rgb, 255:red, 155; green, 155; blue, 155 }  ,fill opacity=0.37 ] (490,120) .. controls (490,114.48) and (494.48,110) .. (500,110) .. controls (505.52,110) and (510,114.48) .. (510,120) .. controls (510,125.52) and (505.52,130) .. (500,130) .. controls (494.48,130) and (490,125.52) .. (490,120) -- cycle ;
\draw  [fill={rgb, 255:red, 155; green, 155; blue, 155 }  ,fill opacity=0.37 ] (280,120) .. controls (280,114.48) and (284.48,110) .. (290,110) .. controls (295.52,110) and (300,114.48) .. (300,120) .. controls (300,125.52) and (295.52,130) .. (290,130) .. controls (284.48,130) and (280,125.52) .. (280,120) -- cycle ;
\draw [fill=gray!7,fill opacity=1]  (250,198) .. controls (250,193.58) and (253.58,190) .. (258,190) -- (322,190) .. controls (326.42,190) and (330,193.58) .. (330,198) -- (330,222) .. controls (330,226.42) and (326.42,230) .. (322,230) -- (258,230) .. controls (253.58,230) and (250,226.42) .. (250,222) -- cycle ;
\draw  [fill={rgb, 255:red, 155; green, 155; blue, 155 }  ,fill opacity=0.37 ] (300,210.5) .. controls (300,204.98) and (304.48,200.5) .. (310,200.5) .. controls (315.52,200.5) and (320,204.98) .. (320,210.5) .. controls (320,216.02) and (315.52,220.5) .. (310,220.5) .. controls (304.48,220.5) and (300,216.02) .. (300,210.5) -- cycle ;
\draw  [fill={rgb, 255:red, 155; green, 155; blue, 155 }  ,fill opacity=0.37 ] (260,210.5) .. controls (260,204.98) and (264.48,200.5) .. (270,200.5) .. controls (275.52,200.5) and (280,204.98) .. (280,210.5) .. controls (280,216.02) and (275.52,220.5) .. (270,220.5) .. controls (264.48,220.5) and (260,216.02) .. (260,210.5) -- cycle ;

\draw  [draw opacity=0][fill={rgb, 255:red, 255; green, 255; blue, 255 }  ,fill opacity=1 ]  (337.5,105) -- (448.5,105) -- (448.5,135) -- (337.5,135) -- cycle  ;
\draw (396,120) node    {$\left( x^{q} -y^{q}\right) \Delta _{q}( x,y)$};
\draw  [draw opacity=0][fill={rgb, 255:red, 255; green, 255; blue, 255 }  ,fill opacity=1 ]  (222.5,149.5) -- (357.5,149.5) -- (357.5,176.5) -- (222.5,176.5) -- cycle  ;
\draw (290,163) node    {$\partial _{x^{q}} H_{q}( x,y) s_{q}( x,y)$};
\draw (500,120) node    {$\tilde{u}$};
\draw (290,120) node    {$q$};
\draw (270.5,210) node    {$u$};
\draw (311,210) node    {$v$};

\end{tikzpicture}}
  \caption{Influence of $x^q$ on $f$. The dashed arrow depicts the noise term, while the solid arrow represents the desired one (from the $(B)$ and $(A)$ terms in \Cref{eq:decompos_intro}, respectively).}
  \label{fig:sub2}
\end{subfigure}
\caption{Dependency graph of a simple \PURECIRC instance and influence graph of the variable $x_q$ on the utility function $f$ defined in \Cref{eq:firstutil}.}
\label{fig:test}
\end{figure}

Computing the partial derivatives of the utility function, we thus have contributions only from $(A)$ and $(B)$. The term $(A)$ represents the desired influence on node $q$. On the other side, the ones coming from $(B)$ are the noise components that introduce an extra dependency on the variables at node $q$. We encapsulate this noise term in a generic term $\Delta_q(x,y)$, which we would like to depend only on $s_q$, mirroring the dependency structure of the \PURECIRC instance. The partial derivatives are thus
\[
\partial_{x^q} f(x,y)={\partial_{x^q} H_q(x,y)} \cdot s_q(x,y)+2(x^q-y^q)\Delta_q(x,y),
\]
and
\[
{\partial_{y^q} f(x,y)=\partial_{y^q} H_q(x,y)}\cdot s_q(x,y)-2(x^q-y^q)\Delta_q(x,y),
\]
where we used that $\partial_{x^q}s_v(x,y)=-\partial_{y^q}s_v(x,y)$ and that $s_v$ depends on $(x^q-y^q)^2$.
In deriving these partial derivatives, we relied on two crucial features of \PURECIRC: (1) $\Delta_q$ contains only derivatives of $s_{\tilde u}$ at a few other nodes $\tilde u$ (those that take $q$ as an input), and (2) $s_q$ does not depend on the variables relative to node $q$. We also refer to \Cref{fig:test} for an illustration.

Now, we would really like the term $|\Delta_q(x,y)|$ to be small. For instance, if we assume $|\Delta_q(x,y)|\ll 1$ and consider the simple linking function $H_q(x,y)=x^q+y^q$, the first-order optimality conditions quickly yield the desired consistency: 
\begin{itemize} 
\item If $s_q(x,y)=1$, then 
\[
\partial_{x^q}f(x,y)(x'-x^q)=(1+2(x^q-y^q)\Delta_q(x,y))(x'-x^q)\le 0\quad\forall x',
\]
implies that $x^q=1$. On the other hand, 
\[
-\partial_{y^q}f(x,y)(y'-y^q)=-(1-2(x^q-y^q)\Delta_q(x,y))(y'-y^q)\le 0\quad\forall y',
\]
implies that $y^q=0$ and that $(x^q-y^q)^2$ is large. 
\item If $s_q(x,y)=0$, then 
\[
\partial_{x^q}f(x,y)(x'-x^q)=2(x^q-y^q)\Delta_q(x,y)(x'-x^q)\le 0\quad\forall x',
\]
and 
\[
-\partial_{y^q}f(x,y)(y'-y^q)=2(x^q-y^q)\Delta_q(x,y)(y'-y^q)\le 0\quad\forall y'.
\]
These two equations imply that $x^q=y^q$ (depending on the sign of $(x^q-y^q)\Delta_q(x,y)$ they are both equal to either $0$ or $1$), and thus $(x^q-y^q)^2$ is small.\footnote{This is true only if $\Delta_q(x,y)\neq0$, but we can ignore this issue for now. This problem will be handled in the final construction by taking copies of the variables.}
\end{itemize} 
This would essentially conclude the reduction. However, it is almost obvious that we cannot guarantee that $|\Delta_q(x,y)|$ is always small. Previous approaches \citep{daskalakis2021complexity, bernasconi2024roleconstraintscomplexityminmax, anagnostides2025complexity} essentially selected $H_q(x,y)=x^q-y^q$, which ensures $\Delta_q$ is small only within an exogenously imposed subset of the hypercube with $x\simeq y$. 

We follow a different approach here using a different function $H_q$ and another way to deal with $\Delta_q$.
We defer for a moment the discussion on the construction of $H_q$ and consistency, and explain the trick we used to deal with $\Delta_q$.

\subsection*{Guessing and Copies}

To overcome the effect of the extra term $\Delta_q$, we use a trick based on repetition (copies) and regularization. We make $n$ copies $x^q=\{x^q_{i}\}_{i\in[n]}$ (and similarly for $y^q$) of the variables at each node $q$, and define the gate functions $s_q$ in terms of $\|x^q-y^q\|^2$. We then add a quadratic regularizer $\varphi(x,y)=\sum_{q\in V}\sum_{i\in [n]} M_i(x_i^q-y_i^q)^2$, where $\{M_i\}_{i\in[n]}$ form a regular grid of $n$ points from $-M$ to $M$ (the range of $M$ will be discussed shortly).  The $M_i$'s act as ``guesses'' for the term $\Delta_q$. Indeed, the partial derivatives  of $f(x,y)=\sum_{q\in V} s_q(x,y)H_q(x,y)+\varphi(x,y)$ now become: 
\begin{subequations}\label{eq:tmp2}
    \begin{equation}
        \partial_{x_i^q}f(x,y)={\partial_{x_i^q} H_q(x,y)} \cdot s_q(x,y)+2(x_i^q-y_i^q)(\Delta_q(x,y)+M_i), 
    \end{equation}
and
    \begin{equation}
        \partial_{y_i^q}f(x,y)={\partial_{y_i^q} H_q(x,y)}\cdot s_q(x,y)-2(x_i^q-y_i^q)(\Delta_q(x,y)+M_i).
    \end{equation}
\end{subequations}
Why does this help? If $|\Delta_q|$ grows slowly enough in $n$ (e.g., sublinearly), there would be at least one $i\in [n]$ (or a few, to be more precise) for which $|\Delta_q+M_i|$ is small, effectively eliminating the $\Delta_q$ term for those copies. However, the more copies $i$ we add, the larger the terms $H_v$, $v \in V$, become (as they depend on more variables), and in turn, so does $\Delta_q$.

\subsection*{Bounding $|\Delta_q|$}

For coordinates $i$ where $|M_i+\Delta_q|$ is large, the dominant term in the gradient (stemming from the regularizer) has different signs for the players and thus aligned objectives. This forces the equilibrium to have $|x^q_i-y^q_i|$ small. In particular, we will prove that at equilibrium, $|x^q_i-y^q_i|=O(|M_i+\Delta_q|^{-1})$, which shows that $\|x^q-y^q\|_1$ is upper bounded by a harmonic series and is thus of order $O(\log n)$. Now we return to bounding $|\Delta_q|$. Up to now, the gadget $H_q$ has been kept as general as possible, as we have used only that it depends on the variables $x^q$ and $y^q$. 
Now, we choose a particular structure for $H_q$ and require the linking function to be of the form $H_q(x,y)=\sum_{i\in[n]}G_i^q(x,y)\cdot(y_i^q-x_i^q)$. With this choice we have, 
\[
|H_q(x,y)|=O(\|y^q-x^q\|_1)=O(\log n).
\]
In particular, since this holds for every node $q\in V$, we can see from \Cref{eq:decompos_intro} that $|\Delta_q|=O(\log n)$ for every $q\in V$, where we ignore the dependency on other parameters. 
Thus, it is sufficient to pick any range of the grid $M$ that grows more than logarithmically in $n$ to ensure that, for many indices $i$, $M_i$ is close to $-\Delta_q$, namely $|\Delta_q+M_i|=O(1)$. We note that we only need to prove a sublinear upper bound for this quantity.

\subsection*{Hiding Consistency Violations in the \LINVI Gadget}

Now, we must show the consistency of the nodes in the circuits, namely that 1) if $s_q(x,y)=1$ then $\|x^q-y^q\|^2$ is large, and that 2) if $s_q(x,y)=0$ then $\|x^q-y^q\|^2$ is small. The second condition holds easily. As shown in \Cref{eq:tmp2}, when $s_q=0$, the surviving part of the partial derivatives of the $i$-th variable would be proportional to $(x_i^q-y_i^q)$ and with opposite signs for the players (and thus the two players have aligned objectives). This forces $x_i^q \simeq y_i^q$ for all $i$ at equilibrium. 

The main challenge is the first condition: if $s_q=1$, we need $\|x_q-y_q\|^2$ to be large. For the ``correctly guessed'' indices $i$ (where $|\Delta_q+M_i|$ is small), the dominant term of the gradients comes from $H_q$. Thus, we have to design a specific gadget $H_q$ to enforce the condition. 
To bound $|\Delta_q|$, we have already constrained $H_q(x,y)=\sum_{i\in[n]}G_i^q(x,y)\cdot(y_i^q-x_i^q)$, so we must now design the $G_i^q$ functions such that $H_q$ has no equilibria when $x_i^q\simeq y_i^q$.
Sadly, by fixed-point arguments, a function $H_q(x,y)$ that is zero for $x^q=y^q$ but has no equilibria at $x^q=y^q$ does not exist.\footnote{We sketch an intuitive proof of this statement here.
Take any $x^q=y^q=z$. Then, since $H_q(z,z)=0$, we have $\partial_{x_i^q}H_q(z,z)=-\partial_{y_i^q}H_q(z,z):=S_i(z)$ for all $i\in[n]$. It follows that any fixed point of the map $z\mapsto\Pi_{[0,1]^n}(z+S(z))$, which are guaranteed to exist by Brouwer's theorem, satisfies the equilibrium conditions.}

To get around this obstacle, we use the main idea of our construction: we hide the symmetric equilibria in a \PPAD-hard problem encoded within the gadget $H_q$. 
If we cannot get a function without symmetric equilibria, we can use a function that has ``hard to find'' (approximately) symmetric equilibria.
More concretely, we allow for violations of the consistency condition, but any violations would lead to a solution of the \PPAD-hard problem hidden in $H_q$.

Remarkably, the desired gadget is the VIs to min-max equilibria reduction already used by \citet{daskalakis2021complexity, bernasconi2024roleconstraintscomplexityminmax}, where, to keep the reduction as simple as possible, we use linear VIs as \citet{bernasconi2024roleconstraintscomplexityminmax}.
However, here we do not have to explicitly add the constraint $x^q_i\simeq y^q_i$, since all the equilibria with $x^q_i$ far from $y^q_i$ satisfy the consistency constraint of \PURECIRC.
This establishes the central dichotomy of our proof. Any solution to our constructed GDA instance must either: (i) be consistent, thus satisfying the $\PURECIRC$ constraints and yielding a solution to $\PURECIRC$; or (ii) be inconsistent for at least one node (i.e., $x^q_i \simeq y^q_i$ even when $s_q=1$), which requires finding a symmetric equilibrium in $H_q$ and thus yields a solution to $\LINVI$. 
Notice that the \LINVI instance must be hidden into any variable $q$ and any copy $i$, and that the VI formulation requires each variable $x^q_i$ to be m-dimensional, where $m$ is the dimension of the \LINVI instance.
Since both \PURECIRC and \LINVI are \PPAD-hard, we conclude that the \GDAFP problem must also be \PPAD-hard.

\begin{remark} In the breakthrough result of \citet{fearnley2021complexity}, the reduction also generated unwanted solutions to a \PPAD-hard problem. In that case, however, the situation was structurally different and much deeper than in our setting. Indeed, they employed a \PLS-hard problem to hide the extra solutions, yielding a solution to either a \PPAD-hard problem or a \PLS-hard problem. In our case, we believe that the reasons for using a ``hiding gadget'' are more pragmatic than fundamental. Ultimately, we reduce from \EOTL, from which we then derive two problems (\PURECIRC and \LINVI).
\end{remark}

\section{Preliminaries}

For $x,y\in\Re^d$, we denote with $\langle x,y\rangle$ the standard Euclidean inner product, and the associated norm as $\|x\|=\sqrt{\langle x,x\rangle }$. The projections $\Pi_D$ on convex and closed sets $D$ are defined according to the $\ell_2$-norm.
We focus on Lipschitz and smooth functions. A continuous and differentiable function is $G$-Lipschitz on $[0,1]^d$ if $|f(x)-f(x')|\le G\|x-x'\|$ for all $x,x'\in [0,1]^d$ and $L$-smooth if $\|\nabla f(x)-\nabla f(x')\|\le L\|x-x'\|$. We denote with $\Re_+$ the set of positive real numbers.

We consider the following total search problem related to finding fixed points of gradient descent-ascent dynamics.

\begin{problem}
[\GDAFP]\label{def:gdafp}
Given $\epsilon$, $L$, $G, B$ $\in\Re_+$, two circuits implementing a $G$-Lipschitz and $L$-smooth function $f:[0,1]^d\times[0,1]^d\to [-B,B]$ and its gradient $\nabla f: [0,1]^d\times[0,1]^d\to\Re^{2d}$, find $(x^\star, y^\star)\in[0,1]^{2d}$ such that 
\[ 
   \frac{\partial f(x^\star,y^\star)}{\partial x_i} ( x_i-x_i^\star)  \le \epsilon \quad \forall i \in [d], x_i \in [0,1],
\]
and
\[ 
   -\frac{\partial f(x^\star,y^\star)}{\partial y_i} ( y_i-y_i^\star) \le \epsilon \quad \forall i \in [d], y_i \in [0,1].
\]
\end{problem}

\begin{remark}In our proof, it will be more convenient to work with the problem \GDAFP defined for individual components. However, the more usual definition requires to find $(x^\star,y^\star)\in[0,1]^{2d}$ such that $\langle \nabla_x f(x^\star,y^\star),x'-x^\star\rangle\le \epsilon$ and $-\langle\nabla_yf(x^\star,y^\star),y'-y^\star\rangle\le \epsilon$ for all $(x',y')\in[0,1]^{2d}$. These definitions are equivalent up to polynomial factors in the approximation.
Moreover, as proved in \citet{daskalakis2021complexity}, this problem is also equivalent to finding (approximate) fixed points of the map $(x,y)\mapsto \left(\Pi_{[0,1]^d}\left(x+\nabla_x f(x,y)\right),\Pi_{[0,1]^d}\left(y-\nabla_y f(x,y)\right)\right)$, justifying the name.
\end{remark}

We will reduce from the following \PPAD-complete problem.

\begin{problem}{(\EOTL)}
    Given two circuits $S,P:\{0,1\}^N\to\{0,1\}^N$ (called, successor and predecessor circuit, respectively), such that $S(0)\neq P(0)=0$, find a node $v\in\{0,1\}^N$ such that $P(S(v))\neq v$ or $S(P(v))\neq v\neq 0$.
\end{problem}

\EOTL is the prototypical problem in \PPAD, in the sense that \PPAD is the class of problems that are reducible to \EOTL. Conversely, a problem is \PPAD-hard if it is as hard as any other problem in \PPAD, i.e., there exists a polynomial-time reduction from \EOTL to that problem. The size of an \EOTL instance is the size of the representation of the predecessor and successor circuit.

We are going to reduce \EOTL to two different intermediate \PPAD-complete problems. The first one is \PURECIRC~\citep{deligkas2022pure}.

\begin{problem}{(\PURECIRC \citep{deligkas2022pure})}
An instance of \PURECIRC is given by a vertex set $V = [\kappa]$ and two sets of gates $\Gcal_{\NOR}$ and $\Gcal_{\PURIFY}$. Each gate is of the form $(u, v, w)$ where $u, v, w \in V$ are distinct nodes with the following interpretation:
\begin{itemize}
\item If $(u,v,w) \in \Gcal_{\NOR}$, then $u$ and $v$ are the inputs of the gate, and $w$ is its output.
\item If $(u,v,w) \in \Gcal_{\PURIFY}$, then $u$ is the input of the gate, and $v$ and $w$ are its outputs.
\end{itemize}

Each node is the output of exactly one gate.
A solution to an instance of \PURECIRC is an assignment $b : V \rightarrow \{0, 1, \bot\}$ that satisfies all the gates,
i.e., for each gate $(u, v, w) \in \Gcal$ we have:
\begin{itemize}
    \item if $(u, v, w) \in \Gcal_{\NOR}$, then $b$ satisfies
    \begin{align*}
        &b(u) = b(v) = 0 \implies b(w) = 1\\
        &(b(u) = 1) \text{  or  } (b(v) = 1) \implies b(w) = 0
    \end{align*}

    \item if $(u, v, w) \in \Gcal_{\PURIFY}$, then $b$ satisfies
    \begin{align*}
        &\{b(v), b(w)\} \cap \{0,1\} \neq \emptyset\\
        &b(u) \in \{0,1\} \implies b(v) = b(w) = b(u)
    \end{align*}
\end{itemize}
\end{problem}
\begin{theorem}[\citep{deligkas2022pure}]
\PURECIRC is \PPAD-complete.
\end{theorem}

The second problem we will use in our reduction is the \LINVI problem introduced by \citet{bernasconi2024roleconstraintscomplexityminmax}.

\begin{problem}[\LINVI\citep{bernasconi2024roleconstraintscomplexityminmax}]
    Given a matrix $D\in[-1,1]^{m\times m}$, a vector $c \in  [-1,1]^m$, and an approximation parameter $\epsilon>0$, find a point $z\in [0,1]^m$ such that:
    \[
    \langle Dz+c,z'-z\rangle \ge -\epsilon\quad\forall z'\in [0,1]^m.
    \]
\end{problem}

In our proof, it will be more convenient to work with the stronger component-wise version of \LINVI, which is proved to be hard by the following theorem.

\begin{theorem}[{\citep[Theorem~4.4]{bernasconi2024roleconstraintscomplexityminmax}}]\label{thm:ppadLinVI}
    There exists a constant $\rho$ such that it is \PPAD-complete to find a $z\in[0,1]^m$ such that 
    \[(Dz+c)_j\cdot (z_j'-z_j)\ge -\rho \quad \forall j \in [m], z'_j \in [0,1]. \]
\end{theorem}

\section{Construction}

We start from an instance of \EOTL. Then, we reduce it to an instance of \LINVI and an instance of \PURECIRC. 
We remark that there is no structural reason to reduce from \EOTL, and we could 
Equivalently, we could reduce directly either from \LINVI or \PURECIRC (and, for instance, in the latter case, construct an instance of \LINVI from one of \PURECIRC), but for clarity, we decided to reduce from the prototypical \PPAD-complete problem \EOTL.

We recall that we let $\kappa$ be the number of vertices of the \PURECIRC instance, $m$ the number of variables of the \LINVI instance, and $\rho$ its approximation. Then, we define the following parameters:

\[
n=\frac{2^{64}m^{14}\kappa^2}{\rho^8}, \quad\epsilon = \frac{\rho^{18}}{2^{140}m^{28}\kappa^4}\quad\text{and}\quad\delta=\frac{\rho^2}{2^{10}m^2}.
\]

Based on the \LINVI and \PURECIRC instances, we build an instance of $\GDAFP$ as follows. We let $M_i=\delta(-n/2+ i)$ for each $i \in [n]$.
We set $d=mn\kappa$, and define the set of variables of the \GDAFP instance as $x^v_{i,j}$ for each $v \in V$, $i \in [n]$, $j \in [m]$. 
For convenience, we denote with $x^v_{i}$ the vector of variables $(x^v_{i,j})_{j \in [m]}$ and by $x^v$ the vector of variables $(x^v_{i,j})_{i\in[n],j\in[m]}$. 
We define 
\[H_v(x,y)=\sum_{i\in [n]} \langle D x^v_i+c,y^v_{i}-x^v_{i}\rangle,\]
where $D$ and $c$ are the $m\times m$-dimensional matrix and $m$-dimensional vector in the \LINVI instance, respectively.

\subsection{Implementing the Gates}
To define the utility function, we will use the following continuous and smooth functions.

We define a function $g:\Re\rightarrow[0,1]$ such that:
\begin{itemize}
\item $g(z)=1$ if $z \le 1/4$
\item $g(z)=0$ if $z\ge 1/2$
\item $g(z)=128(z-1/4)^3-48(z-1/4)^2+1$ if $z\in (1/4,1/2)$.
\end{itemize}
Intuitively, this function can be used to determine the output of a $\NOR$ gate. In particular, given two variables $x,y \in [0,1]$, $g(x+y)$ is $0$, i.e., false, if at least one variable is $1$, while $g(x+y)$ is $1$, i.e., true, if both variables are $0$.

Then, we define the function $\ell:\Re\rightarrow[0,1]$ related to $\PURIFY$ gates as follows:

\begin{itemize}
\item $\ell(z)=0$ if $z\le 5/12$
\item $\ell(z)=1$ if $z\ge 7/12$
\item $\ell(z)=144(z-5/12)^2(2-3z)$ if $z\in (5/12,7/12)$.
\end{itemize}
Given a value $z \in [0,1]$, recall that we interpret every value $z\in (0,1)$ as $\bot$.
Then, considering the functions $\ell(z+ 1/4)$ and $\ell(z-1/4)$, we get that:
\begin{enumerate}[label={\roman*)}]
\item If $z=1$, both $\ell(z+ 1/4)$ and $\ell(z-1/4)$ are $1$
\item If $z=0$, both $\ell(z+ 1/4)$ and $\ell(z-1/4)$ are $0$
\item If $z \in (0,1)$, at least one among $\ell(z+ 1/4)$ and $\ell(z-1/4)$ is in $\{0,1\}$.    
\end{enumerate}

Moreover, we need the following threshold function $\lambda:\Re\rightarrow[0,1]$, whose main property is to be 1 if the input is at least $3m+1$ and $0$ if it is at most $3m$, where we recall that $m$ is the number of variables of the \LINVI instance.

\begin{itemize}
\item $\lambda(z)=0$ if $z\le 3m$
\item $\lambda(z)=1$ if $z\ge 3m+1$
\item {$\lambda(z)=-2(z-3m)^3+3(z-3m)^2$ if $z\in (3m,3m+1)$ }
\end{itemize}
Intuitively, the function $\lambda$ will be used to build the assignment to \PURECIRC. Given a value $z$, the function $\lambda(z)$ assigns value $1$ if $z\ge 3m+1$,  $0$ if $z\le 3m$, and $\bot$ if $z \in (3m,3m+1)$.

Finally, we remark that all the functions are continuous and are $O(1)$-Lipschitz and smooth. 
In the following, we will we use the explicit bounds $\sup_{z\in\Re}|g'(z)|=6$,  $\sup_{z\in\Re}|\ell'(z)|=9$, and $\sup_{m\in\Naturals,z\in\Re}|\lambda'(z)|=3/2$. The three functions are illustrated in \Cref{fig:smoothfnct}.
\begin{figure}
    \centering
    \begin{tikzpicture}
\pgfplotsset{my style/.append style={line width=1.5pt}}
\begin{axis}[
            xtick={-1,0,1/4,1/2,1},
            xticklabels={,$0$,$1/4$,$1/2$,,},
            ytick={0,1},
            yticklabels={$0$,$1$},
            xmin=-0,
            xmax=0.75,
            ymin=-0.1,
            ymax=1.2,
            axis x line=middle, 
            axis y line=middle,
            xlabel={$z$}, 
            x label style={anchor=north},
            ylabel={$g(z)$},
            y label style={anchor=south west},
            width=1.2*144pt,
            height=1.2*89pt,
            axis line style={-Latex[round]}
            ]
\addplot[my style, domain=-7:1/4] {1};
\addplot[my style, domain=1/4:1/2] {128*(x-1/4)^3-48*(x-1/4)^2+1};
\addplot[my style, domain=1/2:0.7] {0};
\end{axis}
\end{tikzpicture}\hfill\begin{tikzpicture}
\pgfplotsset{my style/.append style={line width=1.5pt}}
\begin{axis}[
            xtick={-1,0,5/12,1/2,7/12,1},
            xticklabels={,$0$,$5/12$,,$7/12$,},
            ytick={0,1},
            yticklabels={$0$,$1$},
            xmin=5/12-0.15,
            xmax=7/12+0.15,
            ymin=-0.1,
            ymax=1.2,
            axis x line=middle, 
            axis y line=middle,
            xlabel={$z$}, 
            x label style={anchor=north},
            ylabel={$\ell(z)$},
            y label style={anchor=south west},
            width=1.2*144pt,
            height=1.2*89pt,
            axis line style={-Latex[round]}
            ]
\addplot[my style, domain=-7:5/12] {0};
\addplot[my style, domain=5/12:7/12] {144*(x-5/12)^2*(2-3*x)};
\addplot[my style, domain=7/12:7] {1};
\end{axis}
\end{tikzpicture}\hfill\begin{tikzpicture}
\pgfplotsset{my style/.append style={line width=1.5pt}}
\begin{axis}[
            xtick={0,3,4,5},
            xticklabels={,$3m$,$3m+1$,},
            ytick={0,1},
            yticklabels={$0$,$1$},
            xmin=2.5,
            xmax=4.5,
            ymin=-0.1,
            ymax=1.2,
            axis x line=middle, 
            axis y line=middle,
            xlabel={$z$},
            x label style={anchor=north},
            ylabel={$\lambda(z)$},
            y label style={anchor=south west},
            width=1.2*144pt,
            height=1.2*89pt,
            axis line style={-Latex[round]}
            ]
\addplot[my style, domain=-7:3] {0};
\addplot[my style, domain=3:4] {-2*(x-3)^3+3*(x-3)^2};
\addplot[my style, domain=4:7] {1};
\end{axis}
\end{tikzpicture}
    \caption{Illustration of the smooth functions $g$, $\ell$ and $\lambda$.}
    \label{fig:smoothfnct}
\end{figure}
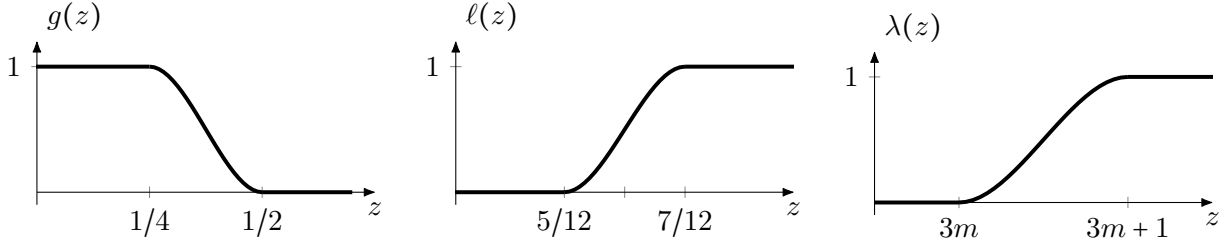

\subsection{Utility Function}
The utility function is composed of three different components.
\begin{itemize}
\item The first one is related to $\NOR$ gates

\begin{align*}
f_1(x,y)=\sum_{\substack{(u,v,w)\in\Gcal_{\NOR}}}  g\left(\lambda(\|x^u-y^u\|^2)+ \lambda(\|x^v-y^v\|^2)\right)\cdot H_w(x,y).
\end{align*}

\item The second one is related to $\PURIFY$ gates

\begin{align*}
f_2(x,y)&= \hspace{-0.2cm}\sum_{(u,v,w)\in\Gcal_{\PURIFY}} \hspace{-0.5cm}\Big[ \ell\Big(\lambda(\|x^u-y^u\|^2) + \frac{1}{4}\Big) \cdot H_v(x,y) 
+ \ell\Big(\lambda(\|x^u-y^u\|^2)- \frac{1}{4}\Big) \cdot H_w(x,y)  \Big].
\end{align*}

\item The last component is a (weighted) regularizer
\begin{align*}
    &\varphi(x,y)=\sum_{v \in V} \sum_{i \in [n]} M_i \lVert x^v_{i}-y^v_{i}\rVert^2.
\end{align*}
Note that the weights $M_i$ can be both positive and negative.
\end{itemize}

Finally, the utility function of \GDA is the sum of the three components
\begin{equation}\label{eq:utilityfunction}
f(x,y)=f_1(x,y)+f_2(x,y)+\varphi(x,y).
\end{equation}

It is easy to see that $f$ is smooth. Moreover, all the parameters of the \GDA instance---specifically, the Lipschitz constant $G$, the range $B$, and the smoothness $L$---are polynomially bounded by the size of the \PURECIRC and \LINVI instances, and thus also by the size of the original \EOTL instance.
Finally, the approximation $\epsilon$ is inversely polynomial in the size of the \EOTL instance.

The intuition behind the construction is the following. 
Given a solution $(x,y)$ to \GDAFP, we build many possible solutions to \LINVI and \PURECIRC. 
First, we check if any of the $x^v_i$, $v \in V$, $i \in [n]$ are solutions to \LINVI. 
If none of them is a solution to \LINVI, then we build an assignment $b(\cdot)$ to \PURECIRC as:

\begin{subequations}\label{eq:interpretation}
\begin{align}
 &\lambda(\lVert x^v-y^v\rVert^2)=1 \implies b(v)=1 \\
 &\lambda(\lVert x^v-y^v\rVert^2)=0 \implies b(v)=0 \\
 &\lambda(\lVert x^v-y^v\rVert^2)\in (0,1) \implies b(v)=\bot
\end{align}
\end{subequations}

Our construction will not guarantee that this assignment is feasible for the \PURECIRC instance; however, it will be so if none of the $\{x_i^v\}_{v\in V, i\in[n]}$ is a solution to \LINVI.

\section{Decomposing the Gradient}

In this section, we compute the partial derivatives of the utility function $f(x,y)$ defined in Equation \eqref{eq:utilityfunction}. This requires first defining two crucial quantities, $\Delta_q(x,y)$ and $s_q(x,y)$, which are helpful to decompose the influence of the circuit structure on the derivatives.

For a given vertex $q\in V$, we define a quantity $\Delta_q(x,y)$ related to the components of the derivative due to gates of which $q$ is an input.
Formally, given any point $(x,y)\in[0,1]^{2d}$, we let
    \begin{align*}
        \Delta_q(x,y)&=\sum_{(q,v,w)\in\Gcal_{\NOR}} g'(\lambda(\|x^q-y^q\|^2)+\lambda(\|x^v-y^v\|^2))\lambda'(\|x^q-y^q\|^2)H_w(x,y)\\
        &+\sum_{(u,q,w)\in\Gcal_{\NOR}} g'(\lambda(\|x^u-y^u\|^2)+\lambda(\|x^q-y^q\|^2))\lambda'(\|x^q-y^q\|^2)H_w(x,y)\\
        &+\sum_{(q,v,w)\in\Gcal_{\PURIFY}} \ell'(\lambda(\|x^q-y^q\|^2)+1/4)\lambda'(\|x^q-y^q\|^2) H_v(x,y)\\
        &+\sum_{(q,v,w)\in\Gcal_{\PURIFY}} \ell'(\lambda(\|x^q-y^q\|^2)-1/4)\lambda'(\|x^q-y^q\|^2) H_w(x,y).
    \end{align*}

Moreover, given a vertex $q$, we define a quantity $s_q(x,y)$.
$s_q$ represents the desired value of the gate whose output is $q$ in a solution to \PURECIRC. Since each node is the output of exactly one gate, only one term in the sum below will be non-zero for any given $q$:
\begin{align*}
    s_q(x,y)&=\sum_{(u,v,q)\in\Gcal_{\NOR}}g(\lambda(\|x^u-y^u\|^2)+\lambda(\|x^v-y^v\|^2))\\
    &+\sum_{(u,q,w)\in \Gcal_{\PURIFY}}\ell(\lambda(\|x^u-y^u\|^2)+1/4)\\
    &+\sum_{(u,v,q)\in \Gcal_{\PURIFY}}\ell(\lambda(\|x^u-y^u\|^2)-1/4).
\end{align*}
For instance, if $q$ is the output of a gate $(u,v,q)\in \Gcal_{\NOR}$, then $s_q$ is the output of the $\NOR$ gate following the interpretation given in \Cref{eq:interpretation}.

The definitions of $\Delta_q$ and $s_q$ allow us to compactly express the gradient of the full utility function $f(x,y)$. The proof of the following lemma is deferred to \Cref{app:grad}.

\begin{restatable}{lemma}{lmGradComp}
The partial derivatives of $f$ (defined as per \Cref{eq:utilityfunction}) are
\begin{subequations}\label{eq:partialDerivatives}
\begin{align}
 \frac{\partial f(x,y)}{\partial {x^{q}_{i,j}}}&=  s_q(x,y) \left[ (D^\top (y^q_i-x^q_i))_j -(D x^q_i+c)_j\right] + 2(M_i+\Delta_q(x,y)) (x^q_{i,j}-y^q_{i,j}),  \\
 \frac{\partial f(x,y)}{\partial {y^{q}_{i,j}}}&=  s_q(x,y) (D x^q_i+c)_j - 2(M_i+\Delta_q(x,y)) (x^q_{i,j}-y^q_{i,j}).
\end{align}
\end{subequations}
\end{restatable}

Intuitively, $s_q(x,y)=1$ ``activates'' the \LINVI part of the gradient, breaking the symmetry between the pseudo-gradients of $x$ and $y$. 
This will force either $x_i^q$ to be an approximate solution to \LINVI for at least one copy $i$, or it will force the distance $\lVert x^q-y^q \rVert^2$ to be large, satisfying the consistency required by \PURECIRC.

The term $\Delta_q(x,y)$ is an undesired byproduct in the gradient of $f(x,y)$. Indeed, this depends on the gates whose inputs include $q$.
Ideally, such gates should not affect the values of the variables $ x^q$ and $y^q$.
The purpose of the weighted regularizer $\varphi$, containing the terms $M_i$, is to ensure that for at least some indices $i$, $M_i+ \Delta_q\simeq 0$, making the impact of $\Delta_q$ on the gradient negligible.

\begin{remark}
The utility function $f(x,y)$ can be expressed as $f(x,y)=\sum_{q\in V}s_q(x,y)H_q(x,y)+\varphi(x,y)$, as anticipated in \Cref{sec:overview}.
Note that the interpretation of $s_q$ is meaningful since each node is the output of exactly one gate.
\end{remark}

\section{\PPAD-Hardness of \GDAFP}

The goal of this section is to show that for any given solution to the $\GDAFP$ instance, we can recover a solution to one of the two $\PPAD$-complete problems: $\PURECIRC$ or $\LINVI$. This implies the $\PPAD$-hardness of $\GDAFP$, as a solution to either problem can be mapped back to a solution for the original $\EOTL$ instance.

The main technical challenge is proving consistency: showing that if no solution $x^v_i$, for $v \in V$, $i \in [n]$,
solves \LINVI, then all the values $s_q$ are consistent with the values $\lambda(\lVert x^q-y^q\rVert^2)$.
Formally, we prove the following dichotomy. \Cref{fig:placeholder} gives a sketch of the argument.

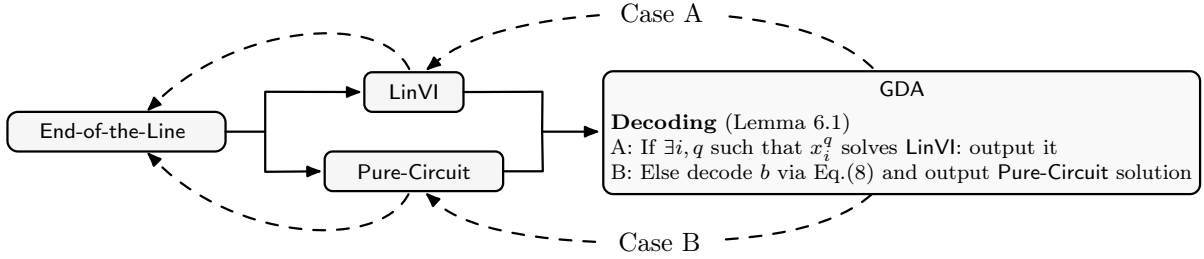
\begin{figure}[t]
    \centering
    \tikzset{
  >=Latex[round],
  every picture/.style={line width=0.85pt,line cap=round,line join=round},
  every node/.style={font=\small},
}

\begin{tikzpicture}[x=0.75pt,y=0.75pt,yscale=-1,xscale=1]

\draw  [->, dash pattern={on 4.5pt off 4.5pt}]  (235,60) .. controls (235.9,17.83) and (149.77,21.74) .. (100.74,69.28) ;
\draw  [->, dash pattern={on 4.5pt off 4.5pt}]  (235,100) .. controls (234.91,140.91) and (149.76,139.79) .. (100.74,90.74) ;
\draw  [fill=gray!7,fill opacity=1] (30,75.41) .. controls (30,72.42) and (32.42,70) .. (35.41,70) -- (134.59,70) .. controls (137.58,70) and (140,72.42) .. (140,75.41) -- (140,84.59) .. controls (140,87.58) and (137.58,90) .. (134.59,90) -- (35.41,90) .. controls (32.42,90) and (30,87.58) .. (30,84.59) -- cycle ;
\draw  [fill=gray!7,fill opacity=1 ] (190,94) .. controls (190,91.79) and (191.79,90) .. (194,90) -- (276,90) .. controls (278.21,90) and (280,91.79) .. (280,94) -- (280,106) .. controls (280,108.21) and (278.21,110) .. (276,110) -- (194,110) .. controls (191.79,110) and (190,108.21) .. (190,106) -- cycle ;
\draw  [fill=gray!7,fill opacity=1 ] (210,54) .. controls (210,51.79) and (211.79,50) .. (214,50) -- (256,50) .. controls (258.21,50) and (260,51.79) .. (260,54) -- (260,66) .. controls (260,68.21) and (258.21,70) .. (256,70) -- (214,70) .. controls (211.79,70) and (210,68.21) .. (210,66) -- cycle ;
\draw  [->]  (260,60) -- (299.84,59.85) -- (299.67,80) -- (328,80) ;
\draw  [->, dash pattern={on 4.5pt off 4.5pt}]  (480,80) .. controls (480.07,20.66) and (379.1,20.22) .. (360,20) .. controls (341.19,19.78) and (273.96,21.72) .. (241.46,48.75) ;
\draw  [->, dash pattern={on 4.5pt off 4.5pt}]  (480,80) .. controls (480.07,139.46) and (375.88,134.53) .. (358.89,135.65) .. controls (342.16,136.76) and (269.14,136.78) .. (241.24,111.19) ;
\draw  [->]  (140,80) -- (160,80) -- (160,59.86) -- (208,59.99) ;
\draw  [->]  (160,80) -- (160,100.26) -- (188,100.02) ;
\draw    (280,100) -- (300,100) -- (299.67,80) ;

\draw (83,79.5) node  [font=\scriptsize] [align=left] {\EOTL};
\draw (235,60) node  [font=\scriptsize] [align=left] {\LINVI};
\draw (235,100) node  [font=\scriptsize] [align=left] {\PURECIRC};
\draw  [fill=gray!5,fill opacity=1 ]  (330.5,54.5) .. controls (330.5,51.74) and (332.74,49.5) .. (335.5,49.5) -- (626.5,49.5) .. controls (629.26,49.5) and (631.5,51.74) .. (631.5,54.5) -- (631.5,106.5) .. controls (631.5,109.26) and (629.26,111.5) .. (626.5,111.5) -- (335.5,111.5) .. controls (332.74,111.5) and (330.5,109.26) .. (330.5,106.5) -- cycle  ;
\draw (481,80.5) node  [font=\scriptsize] [align=left] {\begin{minipage}[lt]{220pt}\setlength\topsep{0pt}
\begin{center}
\GDA
\end{center}
\textbf{Decoding} (\Cref{lm:main})\\
A: If $\exists i,q$ such that $x^q_i$ solves \LINVI: output it\\
B: Else decode $b$ via Eq.\eqref{eq:purecirc} and output \PURECIRC solution
\end{minipage}};
\draw  [draw opacity=0][fill=gray!0,fill opacity=1 ]  (331.5,8) -- (388.5,8) -- (388.5,32) -- (331.5,32) -- cycle  ;
\draw (360,20) node    {Case A};
\draw  [draw opacity=0][fill=gray!0,fill opacity=1 ]  (331.39,123.65) -- (386.39,123.65) -- (386.39,147.65) -- (331.39,147.65) -- cycle  ;
\draw (358.89,135.65) node    {Case B};

\end{tikzpicture}
    \caption{Given an instance $I_{\textnormal EOTL}$ of \EOTL we build an instance $I_{PC}$ of \PURECIRC and $I_{\textnormal{VI}}$ of \LINVI that we combine in an instance $I_{\textnormal GDA}$ of \GDA. From a solution $(x,y)$ of \GDA, we check if any of the $x_i^q$ are a solution to \LINVI; if not, we build an assignment $b:V\to\{0,1,\bot\}$ to \PURECIRC according to \Cref{eq:purecirc}. \Cref{lm:main} assures that either one of the $x_i^q$ is a solution to the \LINVI instance $I_{\textnormal VI}$ (Case A) or that $b$ is a valid assignment to the instance $I_{\textnormal PC}$ of \PURECIRC (Case B). In either case, we can then build a solution to the original \EOTL instance $I_{\textnormal EOTL}$.}
    \label{fig:placeholder}
\end{figure}

\begin{lemma}[Dichotomy Lemma] \label{lm:main}
    Consider a solution $(x,y)\in[0,1]^{2d}$ to the $\GDAFP$ instance. At least one of the two following conditions holds:
    \begin{itemize}
        \item There exists a $q\in V$ and $i \in [n]$ such that $x^q_i$ is a solution to \LINVI
        \item For all vertices $q \in V$, it holds
       \begin{align*}
        &s_{q}(x,y)= 1\implies \lambda(\|x^q-y^q\|^2)=1\\
        &s_{q}(x,y)= 0\implies \lambda(\|x^q-y^q\|^2)=0.
        \end{align*}
    \end{itemize}
\end{lemma}

The previous lemma (whose proof is deferred to \Cref{sec:proofmainlemma}), combined with the definition of $s_q(x,y)$, concludes the proof. In particular, we combine the following results: the functions $g$ and $\ell$ correctly implement the $\NOR$ and $\PURIFY$ gates,
the properties of $s_q(x,y)$, and the interpretation of the \GDA solution in \Cref{eq:interpretation}.

\begin{theorem}
     \GDAFP is \PPAD-hard even for $1/\epsilon$, $L$, $G$ and $B$ polynomial in $d$. 
\end{theorem}

\begin{proof}
Consider any solution $(x,y)$ to \GDAFP.
If there exists a $q\in V$ and $i \in [n]$ such that $x^q_i$ is a solution to \LINVI, then we recover a solution to \LINVI, concluding the proof.

Hence, in the following, we assume that such $x^q_i$ does not exist and build a solution to \PURECIRC. By \Cref{lm:main}, we get that for all vertices $q \in V$:
\begin{subequations}\label{eq:consistency}
        {\begin{align}
        &s_{q}(x,y)= 1\implies \lambda(\|x^q-y^q\|^2)=1\\
        &s_{q}(x,y)= 0\implies \lambda(\|x^q-y^q\|^2)=0.
        \end{align} }
\end{subequations}

We build the assignment $b(\cdot)$ as follows.

\begin{equation}\label{eq:purecirc}
b(v)= \begin{cases}
\lambda(\lVert x^v-y^v\rVert^2) & \text{if } \lambda(\lVert x^v-y^v\rVert^2)\in \{0,1\} \\
\bot & \text{otherwise}
\end{cases}
\end{equation}

In the following, we show that this assignment satisfies all gate constraints.
We start considering the $\NOR$ gates.
 Consider a gate $(u,v,w)\in \Gcal_{\NOR}$. Since $w$ is the output of this $\NOR$ gate, $s_w(x,y)$ simplifies to $s_w(x,y)=g(\lambda(\|x^u-y^u\|^2)+\lambda(\|x^v-y^v\|^2))$.
Then:
\begin{itemize}
    \item If $b(u)=b(v)=0$, then $s_w(x,y)=1$ since $g(0)=1$. Hence, by \Cref{eq:consistency} it holds $b(w)=1$.
    \item If $b(u)$ or $b(v)$ are equal to $1$, then $s_w(x,y)=0$ since $g(q)=0$ for $q\ge 1/2$. Hence, by \Cref{eq:consistency} it holds $b(w)=0$.
    \item Otherwise, any $b(w)\in \{0,1,\bot\}$ satisfies the gate condition.
\end{itemize}

Now, consider the $\PURIFY$ gates.
Consider a gate $(u,v,w)\in \Gcal_{\PURIFY}$. Similarly to $\NOR$ gates, the expression of $s_v(x,y)$ and $s_w(x,y)$ (which are the outputs of the $\PURIFY$ gate), simplifies to $s_v(x,y)=\ell(\lambda(\|x^u-y^u\|^2)+1/4)$ and $s_w(x,y)=\ell(\lambda(\|x^u-y^u\|^2)-1/4)$. Then:
\begin{itemize}
    \item If $b(u)=1$, then $s_{v}(x,y)= 1$ and $s_{w}(x,y)=1$ by the definition of $\ell$. Hence, by \Cref{eq:consistency} it holds $b(v)=b(w)=1$. 
    \item If $b(u)=0$, then $s_{v}(x,y)= 0$ and $s_{w}(x,y)=0$ by the definition of $\ell$. Hence, by \Cref{eq:consistency} it holds $b(v)=b(w)=0$. 
    \item If $b(u)=\bot$, then $s_{v}(x,y)= 1$ or $s_{w}(x,y) = 0$. Indeed, if $s_v(x,y)\neq 1$ then $\lambda(\lVert x^u-y^u\rVert^2)+1/4\le 7/12$ and $s_w(x,y)=0$. On the other hand if $s_w(x,y)\neq 0$ then $\lambda(\lVert x^u-y^u\rVert^2)-1/4\ge 5/12$ and thus $s_v(x,y)=1$.
    Hence, by \Cref{eq:consistency} it holds $b(v)=1$ or $b(w)=0$.
\end{itemize}
This shows that $b(\cdot)$ is a valid assignment to \PURECIRC.
\end{proof}

\section{Proof of \Cref{lm:main}}\label{sec:proofmainlemma}
To prove the main technical lemma, we first provide some auxiliary results regarding the properties of any $\GDAFP$ solution.
In \cref{sec:char}, we argue that for any solution $(x,y)$ of \GDAFP, the distance $|x^q_{i,j}-y^q_{i,j}|$ can be controlled by a function of the other parameters.
This is then used in \Cref{sec:log} to prove a sublinear upper bound of $O(\log n)$ on $|\Delta_q(x,y)|$ in any $\GDAFP$ solution.
In \Cref{sec:many}, we show that, for each node $q$, there are many indices $i\in[n]$ such that $|M_i+\Delta_q|\le 1$ (the ``well-guessed'' copies).
In \Cref{sec:final}, we assume the first condition in \Cref{lm:main} does not hold and prove that the second condition (consistency) must then hold.
Finally, \Cref{sec:endProof} completes the proof.

\subsection{Characterization of $|x_{i,j}^q-y_{i,j}^q|$ for Solutions of \GDAFP} \label{sec:char}

Consider a node $q\in V$, an index $i \in [n]$, and a $j\in [m]$.
We bound the distance between $x^q_{i,j}$ and $y^q_{i,j}$ for any solution to \GDAFP as a function of the gate output $s_q(x,y)$ and the term $|M_i+\Delta_q|$.

\begin{lemma}\label{lm:smallDifference}
    Consider any $q\in V$, $i \in [n]$ and a $j \in [m]$. Then, for any solution $(x,y)\in[0,1]^{2d}$ to \GDAFP, if $M_i+\Delta_q(x,y)\neq 0$, then $(x,y)$ satisfies
    \[ 
    |x^q_{i,j}-y^q_{i,j}|\le \frac{3s_q(x,y)m}{|M_i+\Delta_q(x,y)|}+\sqrt{\frac{\epsilon}{|M_i+\Delta_q(x,y)|}}.
    \]
\end{lemma}

\begin{proof}
    Since $(x,y)$ is a solution to \GDAFP we get that

\[ \frac{\partial f(x,y)}{\partial x^{q}_{i,j}} ( y^q_{i,j}-x^q_{i,j})  \le \epsilon
\quad\text{and}\quad
-\frac{\partial f(x,y)}{\partial y_{i,j}^q} (x^q_{i,j}-y^q_{i,j}) \le \epsilon,\]

where we considered the optimality condition with respect to $y^q_{i,j}$ for the first player and with respect to $x^q_{i,j}$ for the second player.

Computing the values of the partial derivatives according to \Cref{eq:partialDerivatives}, we get

\begin{subequations}
\begin{align}
&\left( s_q(x,y)\left[\left(D^\top(y_i^q-x_i^q)\right)_j-(Dx_i^q+c)_j\right]+2\left(M_i+\Delta_q(x,y)\right) \left(x_{i,j}^q-y_{i,j}^q\right)\right)  \left(  y^q_{i,j}- x^q_{i,j}\right) \le \epsilon \label{eq:KKTx}\\
& \left(-s_q(x,y)\left(Dx_i^q+c\right)_j+2\left(M_i+\Delta_q(x,y)\right)\left(x_{i,j}^q-y_{i,j}^q\right) \right) \left(x^q_{i,j}-y^q_{i,j} \right) \le \epsilon.\label{eq:KKTy}
\end{align}
\end{subequations}
Now we observe that \Cref{eq:KKTx} is of the form $-sd_1z-Mz^2\le \epsilon$, where $s=s_q(x,y)\ge 0$, $d_1=(D^\top(y_i^q-x_i^q))_j-(Dx_i^q+c)_j$, $M=2(M_i+\Delta_q(x,y))$ and $z=(x_{i,j}^q-y_{i,j}^q)$. Similarly, \Cref{eq:KKTy} is of the form $sd_2z+Mz^2\le \epsilon$, where $d_2=-(Dx_i^q+c)_j$. Clearly $|d_1|,|d_2|\le3m$, indeed, by H\"older's inequality,
\begin{align*}
    |d_1|&\le |(D^\top(y_i^q-x_i^q))_j|+|(Dx_i^q+c)_j|\\
    &\le \|D^\top_j\|_1\|y_i^q-x_i^q\|_\infty+\|D_j\|_1\|x_i^q\|_\infty+|c_j|\\
    &\le 3m,
\end{align*}
and similarly for $d_2$.

Now we prove that $|M||z|^2\le \epsilon + 3ms|z|$.
Indeed, 
\begin{itemize}
\item If $M<0$ then the first equation gives $-sd_1z+|M|z^2\le\epsilon$ and thus $|M|z^2\le\epsilon+zsd_1\le \epsilon + 3ms|z|$.
\item  If $M>0$ then the second equation gives $Mz^2\le\epsilon-s d_2 z\le \epsilon+3ms|z|$.
\end{itemize}

Solving the quadratic equation $|M||z|^2\le \epsilon + 3ms|z|$ in $|z|$ gives that 
\begin{align*}
|z|&\le \frac{3ms+\sqrt{9m^2s^2+4|M|\epsilon}}{2|M|}\\
&\le \frac{3ms+\sqrt{9m^2s^2}+\sqrt{4|M|\epsilon}}{2|M|}\\
&\le\frac{3ms}{|M|}+\sqrt{\frac{\epsilon}{|M|}},
\end{align*}
which is the desired bound if we substitute back the values of $s$ and $|M|$. 
Note that the other root of the quadratic equation, namely $\frac{3ms-\sqrt{9m^2s^2+4|M|\epsilon}}{2|M|}$, does not need to be considered as it is negative.
\end{proof}

\subsection{$|\Delta_q(x,y)|$ is Sublinear in $n$} \label{sec:log}

In this section, we show that any solution to $\GDAFP$ guarantees that the undesired component $|\Delta_q(x,y)|$ is sufficiently small for each node $q\in V$. Specifically, we show that it grows at most logarithmically in $n$ (for small enough $\epsilon$). As we will see, this implies that a polynomially large $n$ is sufficient to guarantee that $\max_{i \in [n]} |M_i|> |\Delta_q(x,y)|$. 

We start by providing an upper bound on $\|x^q-y^q\|_1$ (\Cref{lm:norm1}), which is then used to upper-bound $|\Delta_q(x,y)|$ (\Cref{lm:nablaSmall}).

\begin{lemma}\label{lm:norm1}
 For any solution $(x,y)\in[0,1]^{2d}$ of \GDAFP, it holds:
 \[
 \lVert x^q-y^q\rVert_1 \le  14m^2\frac{1}{\delta}\log(n)+mn\sqrt{\frac\epsilon\delta} \quad \forall q \in V.
 \] 
\end{lemma}

\begin{proof}
    Recall that $M_i = (-\frac{n}{2}+ i)\delta$ for all $i\in[n]$ and define $h_0=\frac n2-\frac{\Delta_q(x,y)}\delta$. Then, it follows that $|M_i+\Delta_q(x,y)|=\delta|h_0-i|$. Now, observe that $|M_i+\Delta_q(x,y)|<\delta$ if and only if $|h_0-i|<1$ and that the $i$'s such that $|h_0-i|<1$ are at most $2$.
    We also define
    \[
    S_{\ge\delta}=\sum_{i\in[n]: |M_i+\Delta_q(x,y)|\ge\delta}\frac{1}{|M_i+\Delta_q(x,y)|} = \frac{1}{\delta}\sum_{i\in[n]:|h_0-i|\ge1}\frac{1}{|h_0-i|}.
    \]
    
    Finally, we define a bucket $B_k=\left\{i\in[n]:|i-h_0|\in[k,k+1)\right\}$ for all $k\in\mathbb{N}$.
    Then, for all $k\in\mathbb{N}$ we have that $|B_k|\le 2$ and for each $i\in B_k$ it holds
    \[
    \frac{1}{|M_i+\Delta_q(x,y)|}=\frac{1}{\delta|h_0-i|}\le\frac{1}{\delta k}\quad\forall i\in B_k.
    \]

    Clearly, there are at most $n$ non-empty buckets that cover $[n]$. More formally, there exists a $j\ge 0$ such that all $i$'s such that $|i-h_0|\ge 1$ are contained in $\bigcup_{k\in[n]} B_{k+j}$. Thus, the following inequalities hold
    \begin{align*}
        S_{\ge\delta} \le\frac1\delta\sum_{k=1}^n\sum_{i\in B_{k+j}}\frac{1}{|h_0-i|}
        \le \frac{2}{\delta}\sum_{k=1}^n\frac{1}{k+j}
        \le \frac{2}{\delta}\sum_{k=1}^n\frac{1}{k}.
    \end{align*}
    Hence, $S_{\ge\delta}\le \frac{4}{\delta}\log(n)$ where we used the upper bound $2\log(n)$ of the harmonic series.
    Now, we bound $\|x^q-y^q\|_1$ as follows
    \begin{align*}
        \|x^q-y^q\|_1&=\sum_{i\in[n]}\|x^q_i-y^q_i\|_1\\
        &=\sum_{i\in[n]: |M_i+\Delta_q(x,y)|\ge \delta} \|x_i^q-y_i^q\|_1+\sum_{i\in[n]: |M_i+\Delta_q(x,y)|< \delta} \|x_i^q-y_i^q\|_1\\
        &\le \sum_{i\in[n]: |M_i+\Delta_q(x,y)|\ge \delta} \|x_i^q-y_i^q\|_1+2m\tag{$\|x_i^q-y_i^q\|_1\le m$}\\
        &\le \sum_{i\in[n]:|M_i+\Delta_q(x,y)|\ge \delta}\sum_{j\in[m]} \left[\frac{3s_q(x,y)m}{|M_i+\Delta_q(x,y)|}+\sqrt{\frac{\epsilon}{|M_i+\Delta_q(x,y)|}}\right]+2m\tag{\Cref{lm:smallDifference}}\\
        &\le 3s_q(x,y)m^2S_{\ge\delta}+mn\sqrt{\frac\epsilon\delta}+2m\\
        &\le 12s_q(x,y)m^2 \frac{1}{\delta}\log(n)+mn\sqrt{\frac\epsilon\delta}+2m\\
        &\le 14m^2 \frac{1}{\delta}\log\left(n\right)+mn\sqrt{\frac\epsilon\delta}
    \end{align*}
    concluding the proof.
\end{proof}

The previous lemma implies that $|\Delta_q(x,y)|$ grows at most logarithmically in $n$ when $\epsilon$ is sufficiently small. 

\begin{lemma}\label{lm:nablaSmall}
     For any solution $(x,y)\in[0,1]^{2d}$ of \GDAFP, it holds:
    \[ 
    |\Delta_{q}(x,y)| \le 2^9\frac{m^3\kappa}{\delta}\log(n)+2^5m^2n \kappa \sqrt{\frac\epsilon\delta} \quad \forall q \in V.
    \] 
\end{lemma}

\begin{proof}
We observe that $q$ is the input of at most $\kappa$ gates and that the derivatives of $g$, $\lambda$, and $\ell$ are constant with absolute values upper bounded by $6$, $3/2$, and $9$, respectively. Moreover, $|\langle D x^w_i+c,y^w_{i}-x^w_{i}\rangle|\le \|Dx_i^w+c\|_\infty\|y_i^w-x_i^w\|_1\le 2m\|x_i^w-y_i^w\|_1$ by H\"older's inequality. This implies that $|H_w(x,y)|\le 2m\|x^w-y^w\|_1$.
Thus
\[
|\Delta_q(x,y)|\le 14\kappa \max_{w\in V}|H_w(x,y)|\le 28m\kappa\left(14m^2\frac1\delta\log(n)+mn\sqrt{\frac\epsilon\delta}\right)\le2^9\frac{m^3\kappa}{\delta}\log(n)+2^5m^2n\kappa \sqrt{\frac\epsilon\delta}, 
\]
where in the first inequality we bounded the product of the two derivatives ($g$ and $\lambda$ or $\ell$ and $\lambda$) with $9\cdot \frac{3}{2}\le14$ and we used that there are at most $\kappa$ gates, while in the second inequality we used \Cref{lm:norm1}.
\end{proof}

\subsection{$|\Delta_q(x,y)+M_i|\le 1$ for $\Omega(1/\delta)$ Many $i$'s} \label{sec:many}

Consider a node $q\in V$. The result from \cref{lm:nablaSmall} shows that the undesired term $|\Delta_q(x,y)|$ grows at most logarithmically in $n$ . Then, by setting $n$ to be polynomially large in the instance size, we ensure that the range of the regularizer terms $M_i$ (which spans approximately $\delta n$) exceeds $|\Delta_q|$. This guarantees that there are $\Omega({1}/{\delta})$ indices $i\in [n]$ such that $|\Delta_q+M_i|\le 1$. These are the ``well-guessed'' copies where the regularizer $M_i$ neutralizes the undesired component $\Delta_q$.

\begin{lemma}\label{lm:manyCloseToZero}
   For any solution $(x,y)\in [0,1]^{2d}$ to \GDAFP, there exists at least $1/\delta$ indices $i\in[n]$ such that $|\Delta_q(x,y)+M_i|\le 1$.
\end{lemma}

\begin{proof}

First, we show that our choice of the parameters $n$, $\delta$, and $\epsilon$ implies that  $|\Delta_q(x,y)|\le n\delta/4$. In particular, it is easy to verify that our choice of $\delta,n,\epsilon$ implies $2^{24}\frac{m^6\kappa^2}{\delta^4}\le n$ and $\epsilon\le\frac{\delta^3}{2^{16}m^4\kappa^2}$. By rearranging the first one, we get $\frac{2^9m^3\kappa}{\delta}\le\frac{\sqrt{n}\delta}{8}$ which gives:
\begin{align*}
\frac{2^9m^3\kappa}{\delta}\sqrt{n}\le\frac{n\delta}{8},\tag{A}
\end{align*}
while the second one gives:
\begin{align*}
2^5m^2n\kappa\sqrt{\frac\epsilon\delta}\le \frac{2^5m^2n\kappa}{\sqrt\delta}\frac{\delta^{3/2}}{2^8m^2\kappa}=\frac{\delta n}{8}.\tag{B}
\end{align*}
We can then use \Cref{lm:nablaSmall} to upper bound $|\Delta_q(x,y)|$ as follows:
\begin{align*}
    |\Delta_q(x,y)|&\le 2^9\frac{m^3\kappa}{\delta}\log(n)+2^5m^2n \kappa \sqrt{\frac\epsilon\delta}\tag{from \Cref{lm:nablaSmall}}\\
    &\le 2^9\frac{m^3\kappa}{\delta}\sqrt{n}+2^5m^2n \kappa \sqrt{\frac\epsilon\delta}\\
    &\le \frac{n\delta}{4}.\tag{by $(A)$ and (B)}
\end{align*}

    Now, defining $h_0=\frac{n}{2}-\frac{\Delta_q(x,y)}{\delta}$, we get that $\frac{1}{4}n\le h_0\le \frac{3}{4}n$, and $|M_i+\Delta_q(x,y)|=\delta|h_0-i|\le1$ if $|i-h_0|\le\frac{1}{\delta}$. 
Moreover, we get that there exists a $j \in [n/4,3n/4]$ such that $|j-h_0|\le 1$.
Since, for our choice of parameters, it holds that $\frac{1}{\delta}\le n/4-1$, we also get that the set of indices
\[
\mathcal{I}=\left\{i\in\Naturals:i=j+k, |k|\le \frac{1}{\delta}, k \in \mathbb{Z}\right\}
\]
is a subset of $[n]$.

Recall that, for each $i\in \mathcal{I}$, it holds $|i-h_0|\le \frac{1}{\delta}$ and hence $|M_i+\Delta_q(x,y)|\le 1$.
Noticing that $|\mathcal{I}|\ge 1/\delta$, we can conclude that there are at least $1/\delta$ many indices that guarantee $|\Delta_q(x,y)+M_i|\le 1$.
\end{proof}

\subsection{Consistency: Relating $s_q$ to $\lambda(\lVert x^q-y^q\rVert^2)$ }\label{sec:final}%

Now, we are finally ready to prove some results about consistency.
First, we show that if $s_q(x,y)=0$, then $\lVert x^q-y^q \rVert^2\le 3m$.
Notice that this result is independent of the first condition of \Cref{lm:main} and the existence of approximate \LINVI solutions.
This result directly follows from \Cref{lm:smallDifference}. Indeed, since $s_q(x,y)=0$,  $\lVert x^q-y^q\rVert^2$ is small for each $i$ such that $|M_i+\Delta_q(x,y)|\gg 0$.

\begin{lemma}\label{lm:sEqual0}
    Consider a solution $(x,y)\in[0,1]^{2d}$ to \GDAFP and a node $q$ such that $s_q(x,y)=0$. Then, $ \lVert x^q-y^q\rVert^2\le 3m$. 
\end{lemma}

\begin{proof}
We observe that $|M_i+\Delta_q|<\delta$ for at most two indices $i$ and for those indices $\lVert x_{i}^q-y_{i}^q\rVert^2\le m$. For all the other indices $i$ and any $j \in [m]$, we get $| x_{i,j}^q-y_{i,j}^q|\le \sqrt{\frac{\epsilon}{|M_i+\Delta_q(x,y)|}} \le \sqrt{\frac{\epsilon}{\delta}}$ from \Cref{lm:smallDifference} since $s_q(x,y)=0$.
Thus,
\begin{align*}    
\|x^q_i-y^q_i\|^2&=\sum_{j\in[m]}|x_{i,j}^q-y_{i,j}^q|^2\\
&\le \sum_{j\in[m]}\frac{\epsilon}{\delta}=m\frac{\epsilon}{\delta}.
\end{align*}

 Combining the results for the indices such that $|M_i+\Delta_q|<\delta$ and $|M_i+\Delta_q|\ge\delta$ we get
\[
\|x^q-y^q\|^2\le2m+nm\frac{\epsilon}{\delta}\le3m,
\]
where we use that $\epsilon\le\delta/n$ by our choice of the parameters.
\end{proof}

To conclude the proof, we consider the case in which $s_q(x,y)=1$. Here, it is necessary to use the first condition of \Cref{lm:main}.
In particular, we prove that $\lVert x^q-y^q\rVert^2$ is large (at least $3m+1$) or that we find a solution to \LINVI.

\begin{lemma}\label{lm:functionOfEta}
Consider any solution $(x,y)\in[0,1]^{2d}$ to \GDAFP and any node $q\in V$ such that $s_{q}(x,y)=1$. If that there are no $i\in[n]$, such that $x^q_i$ is a solution to \LINVI, then
\( 
\lVert x^q-y^q\rVert^2 \ge 3m+1.
\)
\end{lemma}

\begin{proof}
    From \Cref{lm:manyCloseToZero}, for at least $1/\delta$ many $i\in[n]$, we have that $|M_i+\Delta_q(x,y)|\le 1$. 
    We start proving that for such $i$'s it holds $\|x_i^q-y_i^q\|_1\ge \rho/10$. 
    
    Assume by contradiction that for some such $i$'s we have $\|x_i^q-y_i^q\|_1< \frac\rho{10}$.
    Then, given that $(x,y)$ is a solution to \GDAFP, and by assumption $s_q(x,y)=1$, the optimality conditions of the first player require that for all $j\in[m]$ and $z\in[0,1]$:
    \[
    \left( (D^\top(y_i^q-x_i^q))_j-(Dx_i^q+c)_j+2(M_i+\Delta_q(x,y)) (x_{i,j}^q-y_{i,j}^q)\right)  (  z- x^q_{i,j}) \le \epsilon,
    \]
   
    which implies that for all $j\in[m]$ and $z\in[0,1]$
    \begin{align*}        
    -(D x^q_i+c)_j(  z- x^q_{i,j})&\le \epsilon+(D^\top(x_i^q-y_i^q))_j(z-x_{i,j}^q)+2(M_i+\Delta_q(x,y))(y_{i,j}^q-x_{i,j}^q)(z-x_{i,j}^q)\\
    &\le\epsilon + \frac{\rho}{10}+2\frac\rho{10}\le \rho,
    \end{align*}
    where we used $\epsilon\le\rho/2$. This is a contradiction, since this shows that $x_{i}^q$ is a solution to \LINVI.
    Thus, there are at least $1/\delta$ many $i$'s such that $\|x_i^q-y_i^q\|_1\ge\rho/10$.
    Hence, it holds
    \[
    \|x^q-y^q\|^2\ge\sum_{i\in[n]:|M_i+\Delta_q(x,y)|\le 1}\|x_i^q-y_i^q\|^2\ge \frac{1}{m} \sum_{i\in[n]:|M_i+\Delta_q(x,y)|\le 1}\|x_i^q-y_i^q\|_1^2 \ge  \frac{\rho^2}{100m\delta},
    \]
    where in the second to last inequality we used that $\|\cdot\|^2_1\le m\|\cdot\|^2$.
    The proof is concluded by observing that $\delta\le \frac{\rho^2}{2^9m^2}$ and hence $\frac{\rho^2}{100m\delta}\ge 5m$.
\end{proof}

\subsection{Piecing Everything Together} \label{sec:endProof}

\Cref{lm:main} directly follows from \Cref{lm:sEqual0}, \Cref{lm:functionOfEta} and the definition of the function $\lambda$.
In particular, if there exists a $q\in V$ and $i \in [n]$ such that $x^q_i$ is a solution to \LINVI, then the first condition of the lemma is satisfied.

Otherwise, the following holds:
\begin{itemize}
\item if $s_q(x,y)=0$, by \Cref{lm:sEqual0} it holds $\lVert x^q-y^q\rVert^2 \le 3m$ and,  by the definition of $\lambda$, we have $\lambda(\lVert x^q-y^q\rVert^2)=0$.
\item if $s_q(x,y)=1$, by \Cref{lm:functionOfEta} it holds $\lVert x^q-y^q\rVert^2 \ge 3m+1$ and,  by the definition of $\lambda$, we have $\lambda(\lVert x^q-y^q\rVert^2)=1$.
\end{itemize}
This concludes the proof of \Cref{lm:main}.

\section{Conclusion and Open Problems}

In this paper, we showed that min-max optimization with product constraints is \PPAD-hard even when all parameters are polynomially bounded. It remains an open question whether stronger inapproximability results exist, for instance, when all parameters except the dimensionality are constants.
Moreover, our reduction is based on white-box access to the \PURECIRC and \LINVI instances. On the other hand, \citet{daskalakis2021complexity}, used a reduction that had only black-box access to a Brouwer instance, and thus could use the unconditional lower bounds of \citet{hirsch1989exponential} to prove a lower bound on the number of queries to the gradients. Our current techniques do not preclude a time-inefficient but query-efficient algorithm (for instance, \citet{chen2024computing} recently showed a query-efficient algorithm for $\ell_\infty$ contractions).
It would be interesting to improve our reduction to exploit only black-box access to any \PPAD-complete problem for which there are known query lower bounds.
Lastly, the complexity of finding a solution in constant dimension (e.g., $d=2$) remains an intriguing open problem. Our current reduction relies on a high-dimensional construction and does not directly extend to low-dimensional settings.

\newpage
\appendix

\section*{Appendix}

\section{Omitted Proofs}\label{app:grad}

\lmGradComp*

\begin{proof}

We recall that $H_v(x,y)=\sum_{i\in[n]} \langle Dx_i^v+c,y_i^v-x_i^v\rangle$ for all $v\in V$. Thus, 
\[
\frac{\partial H_v(x,y)}{\partial x_{i,j}^v} = (D^\top(y_i^v-x_i^v))_j-(Dx_i^v+c)_j
\quad\text{and}\quad
\frac{\partial H_v(x,y)}{\partial y_{i,j}^v} = (Dx_i^v+c)_j.
\]

We are going to use the (composite) chain rule. Formally, given three functions $t,r,u:\Re\to\Re$, it holds:
\[
[t(r(u(x)))]'=t'(r(u(x)))r'(u(x))u'(x).
\]

\paragraph{Player $x$:}
First, let us compute the partial derivative of $f_1$ with respect to $x_{i,j}^q$:
    \begin{align*}
        \frac{\partial f_1(x,y)}{\partial x_{i,j}^q}%
        &=2(x_{i,j}^q-y_{i,j}^q)\sum_{(q,v,w)\in\Gcal_{\NOR}} g'(\lambda(\|x^q-y^q\|^2)+\lambda(\|x^v-y^v\|^2))\lambda'(\|x^q-y^q\|^2)H_w(x,y)\\
        &+2(x_{i,j}^q-y_{i,j}^q)\sum_{(u,q,w)\in\Gcal_{\NOR}} g'(\lambda(\|x^u-y^u\|^2)+\lambda(\|x^q-y^q\|^2))\lambda'(\|x^q-y^q\|^2)H_w(x,y)\\
        &+\left[(D^\top(y_i^q-x_i^q))_j-(Dx_i^q+c)_j\right]\sum_{(u,v,q)\in\Gcal_{\NOR}}g(\lambda(\|x^u-y^u\|^2)+\lambda(\|x^v-y^v\|^2)).
    \end{align*}
Then, we compute the partial derivative of $f_2$ with respect to $x_{i,j}^q$:
    \begin{align*}
        \frac{\partial f_2(x,y)}{\partial x_{i,j}^q} %
        &=2(x_{i,j}^q-y_{i,j}^q)\sum_{(q,v,w)\in\Gcal_{\PURIFY}} \ell'(\lambda(\|x^q-y^q\|^2)+1/4)\lambda'(\|x^q-y^q\|^2) H_v(x,y)\\
        &+2(x_{i,j}^q-y_{i,j}^q)\sum_{(q,v,w)\in\Gcal_{\PURIFY}} \ell'(\lambda(\|x^q-y^q\|^2)-1/4)\lambda'(\|x^q-y^q\|^2) H_w(x,y)\\
        &+\left[(D^\top(y_i^q-x_i^q))_j-(Dx_i^q+c)_j\right]\sum_{(u,q,w)\in \Gcal_{\PURIFY}}\ell(\lambda(\|x^u-y^u\|^2)+1/4)\\
        &+\left[(D^\top(y_i^q-x_i^q))_j-(Dx_i^q+c)_j\right]\sum_{(u,v,q)\in \Gcal_{\PURIFY}}\ell(\lambda(\|x^u-y^u\|^2)-1/4),
    \end{align*}
and finally the partial derivative of $\varphi$ with respect to $x^q_{i,j}$:
    \begin{align*}
        \frac{\partial\varphi(x,y)}{\partial x_{i,j}^q} &=2M_i(x_{i,j}^q-y_{i,j}^q).
    \end{align*}

\paragraph{Player $y$:}
Similarly, the partial derivative of $f_1$ with respect to $y_{i,j}^q$ is

\begin{align*}
    \frac{\partial f_1(x,y)}{\partial y_{i,j}^q}&=2(y_{i,j}^q-x_{i,j}^q)\sum_{(q,v,w)\in\Gcal_{\NOR}} g'(\lambda(\|x^q-y^q\|^2)+\lambda(\|x^v-y^v\|^2))\lambda'(\|x^q-y^q\|^2)H_w(x,y)\\
    &+2(y_{i,j}^q-x_{i,j}^q)\sum_{(u,q,w)\in\Gcal_{\NOR}}g'(\lambda(\|x^u-y^u\|^2)+\lambda(\|x^q-y^q\|^2))\lambda'(\|x^q-y^q\|^2)H_w(x,y)\\
    &+(Dx_i^q+c)_j\sum_{(u,v,q)\in\Gcal_{\NOR}}g(\lambda(\|x^u-y^u\|^2)+\lambda(\|x^v-y^v\|^2)).
\end{align*}

The partial derivative of $f_2$ with respect to $y_{i,j}^q$ is
\begin{align*}
    \frac{\partial f_2(x,y)}{\partial y^q_{i,j}}&=2(y_{i,j}^q-x_{i,j}^q)\sum_{(q,v,w)\in\Gcal_{\PURIFY}}\ell'(\lambda(\|x^q-y^q\|^2)+1/4)\lambda'(\|x^q-y^q\|^2)H_v(x,y)\\
    &+2(y_{i,j}^q-x_{i,j}^q)\sum_{(q,v,w)\in\Gcal_{\PURIFY}}\ell'(\lambda(\|x^q-y^q\|^2)-1/4)\lambda'(\|x^q-y^q\|^2)H_w(x,y)\\
    &+(Dx_i^q+c)_j\sum_{(u,q,w)\in\Gcal_{\PURIFY}}\ell(\lambda(\|x^u-y^u\|^2)+1/4)\\
    &+(Dx_i^q+c)_j\sum_{(u,v,q)\in\Gcal_{\PURIFY}}\ell(\lambda(\|x^u-y^u\|^2)-1/4).
\end{align*}

Finally, the partial derivative of $\phi$ with respect to $y^{q}_{i,j}$ is
\[
\frac{\partial \varphi(x,y)}{\partial y_{i,j}^q}= 2M_i(y_{i,j}^q-x_{i,j}^q).
\]

\paragraph{Combining the Terms:}
By defining
    \begin{align*}
        \Delta_q(x,y)&=\sum_{(q,v,w)\in\Gcal_{\NOR}} g'(\lambda(\|x^q-y^q\|^2)+\lambda(\|x^v-y^v\|^2))\lambda'(\|x^q-y^q\|^2)H_w(x,y)\\
        &+\sum_{(u,q,w)\in\Gcal_{\NOR}} g'(\lambda(\|x^u-y^u\|^2)+\lambda(\|x^q-y^q\|^2))\lambda'(\|x^q-y^q\|^2)H_w(x,y)\\
        &+\sum_{(q,v,w)\in\Gcal_{\PURIFY}} \ell'(\lambda(\|x^q-y^q\|^2)+1/4)\lambda'(\|x^q-y^q\|^2) H_v(x,y)\\
        &+\sum_{(q,v,w)\in\Gcal_{\PURIFY}} \ell'(\lambda(\|x^q-y^q\|^2)-1/4)\lambda'(\|x^q-y^q\|^2) H_w(x,y),
    \end{align*}
    and 
    \begin{align*}
        s_q(x,y)&=\sum_{(u,v,q)\in\Gcal_{\NOR}}g(\lambda(\|x^u-y^u\|^2)+\lambda(\|x^v-y^v\|^2))\\
        &+\sum_{(u,q,w)\in \Gcal_{\PURIFY}}\ell(\lambda(\|x^u-y^u\|^2)+1/4)\\
        &+\sum_{(u,v,q)\in \Gcal_{\PURIFY}}\ell(\lambda(\|x^u-y^u\|^2)-1/4),
    \end{align*}

    we can collect all the appropriate terms in $\partial_{x_{i,j}^q}f$ and $\partial_{y_{i,j}^q}f$ leading to
    \[
    \frac{\partial f(x,y)}{\partial x_{i,j}^q} =2(x_{i,j}^q-y_{i,j}^q)(M_i+\Delta_q(x,y))+\left[(D^\top(y_i^q-x_i^q))_j-(Dx_i^q+c)_j\right]s_q(x,y)
    \]
    and 
    \[
    \frac{\partial f(x,y)}{\partial y_{i,j}^q}=2(y_{i,j}^q-x_{i,j}^q)(M_i+\Delta_q(x,y))+(Dx_i^q+c)_js_q(x,y)
    \]
    as desired.
\end{proof}

\section*{Acknowledgments}
The work of MB was partially funded by the European Union. Views and opinions expressed
are however those of the author(s) only and do not necessarily reflect those of the European Union or the
European Research Council Executive Agency. Neither the European Union nor the granting authority can
be held responsible for them.

This work is supported by an ERC grant (Project 101165466 — PLA-STEER), by the FAIR (Future Artificial Intelligence Research) project PE0000013, funded by the NextGenerationEU program within the PNRRPE-AI scheme (M4C2, Investment 1.3, Line on Artificial Intelligence), and by the EU Horizon project ELIAS (European Lighthouse of AI for Sustainability, No. 101120237).

We thank Andrea Celli, Gabriele Farina, and the anonymous reviewers for their helpful comments and discussions. 
\newpage
\printbibliography

\end{document}